\documentclass[10pt]{article}
\usepackage[utf8]{inputenc}
\usepackage[T1]{fontenc}
\usepackage{amsmath,amsthm,amsfonts,amssymb}
\usepackage{graphicx}
\usepackage{physics}
\usepackage{tikz}
\usepackage{siunitx}
\usetikzlibrary{arrows.meta, positioning, calc, decorations.pathreplacing, patterns}
\usepackage{pgfplots}
\pgfplotsset{compat=1.18}
\usepackage{enumitem}
\usepackage{hyperref}
\usepackage{tikz-cd}
\usepackage{geometry}
\newtheorem{theorem}{Theorem}[section]
\newtheorem{lemma}[theorem]{Lemma}
\newtheorem{proposition}[theorem]{Proposition}
\newtheorem{corollary}[theorem]{Corollary}
\newtheorem{definition}[theorem]{Definition}
\newtheorem{assumption}[theorem]{Assumption}
\newtheorem{remark}[theorem]{Remark}

\newcommand{\cB}{\mathcal{B}}

\newcommand{\cD}{\mathcal{D}}

\newcommand{\cH}{\mathcal{H}}

\newcommand{\cL}{\mathcal{L}}

\newcommand{\cP}{\mathcal{P}}
\newcommand{\cQ}{\mathcal{Q}}

\newcommand{\cT}{\mathcal{T}}

\newcommand{\C}{\mathbb{C}}
\newcommand{\N}{\mathbb{N}}

\newcommand{\id}{\mathbb{I}}

\newcommand{\poly}{\operatorname{poly}}
\newcommand{\eps}{\varepsilon}

\newcommand{\spec}{\operatorname{spec}}
\newcommand{\re}{\operatorname{Re}}

\newcommand{\ran}{\operatorname{ran}}

\newcommand{\dnorm}[1]{\left\Vert#1\right\Vert_{\diamond}}

\AtBeginDocument{\RenewCommandCopy\qty\SI}

\title{Operator-Algebraic Methods for Asymptotic-Preserving Quantum Simulation of Open Systems}
\author{M.W. AlMasri\thanks{Wilczek Quantum Center, School of Physics and Astronomy, Shanghai Jiao Tong University, Minhang, Shanghai, China, Email: mwalmasri2003@gmail.com}}
\date{}

\begin{document}
\maketitle

\begin{abstract}
\noindent We develop a mathematically rigorous framework for simulating \emph{multiscale physical systems} using quantum computational resources, by translating the \emph{language of asymptotic-preserving (AP) schemes} into the formalism of quantum channels and Lindbladian dynamics. For stiff open quantum systems governed by singularly perturbed generators $\cL_\eps = \eps^{-1}\cL_{\mathrm{fast}} + \cL_{\mathrm{slow}}$ with $\eps \to 0$, we prove that a protocol starting in the slow manifold (with the general state case sketched via an initial boundary layer) followed by effective evolution converges uniformly to the reduced slow state of the true stiff dynamics, with an explicit global error bound $\mathcal{O}(\eps^2(1+T) + T\Delta t^p)$ that rigorously captures $\mathcal{O}(\eps)$ physical effects (e.g., Purcell decay) without unphysical dissipation. We establish resource-complexity scaling showing that avoiding the $\Omega(1/\eps)$ stiffness penalty while capturing $\mathcal{O}(\eps)$ dynamics is efficiently achieved via analytic adiabatic elimination to precompute the effective generator, providing dimensional savings amortized against classical precomputation. The framework is illustrated through cavity QED in the bad-cavity limit and a heuristic quantum-inspired AP discretization of linear kinetic equations. This work provides a principled mathematical bridge between classical multiscale numerical analysis and quantum simulation algorithms.
\end{abstract}

\section{Introduction}
\label{sec:introduction}

The simulation of multiscale physical systems, spanning disparate temporal, spatial, or energetic scales, represents a grand challenge across computational physics, chemistry, and engineering. Classical numerical analysis addresses such problems through asymptotic-preserving (AP) schemes~\cite{Jin0,Jin,e2007hmm,E,Hairer}, which maintain uniform accuracy as scale-separation parameters vanish, thereby avoiding prohibitive resolution requirements. Concurrently, quantum simulation offers exponential potential speedups for certain many-body and open-system dynamics~\cite{Georgescu,Cirstoiu}. However, a systematic methodology for importing the language of AP schemes---commutative diagrams, limiting procedures, uniform error bounds---into the quantum computational setting remains underdeveloped.

In this work, we bridge these domains by formulating AP quantum simulation as a rigorous operator-theoretic framework for multiscale physical systems. We consider stiff open quantum systems whose Lindbladian generators admit singular perturbations:
\begin{equation}
    \mathcal{L}_\eps = \frac{1}{\eps}\mathcal{L}_{\mathrm{fast}} + \mathcal{L}_{\mathrm{slow}}, \qquad 0 < \eps \ll 1,
    \label{eq:singular_perturbation}
\end{equation}
where $\mathcal{L}_{\mathrm{fast}}$ generates rapid dissipative relaxation on time scale $\mathcal{O}(\eps)$ and $\mathcal{L}_{\mathrm{slow}}$ governs slow coherent or dissipative dynamics on time scale $\mathcal{O}(1)$. Standard quantum simulation algorithms based on naive Trotter--Suzuki product formulas~\cite{Childs,Lloyd1996,Suzuki1990} require time steps $\Delta t = \mathcal{O}(\eps)$ to resolve the fast scale, leading to gate complexity $\Omega(T/\eps)$ even when observables of interest evolve slowly. This stiffness penalty can render simulations infeasible on near-term hardware.

Our central insight is that the language of AP schemes can be faithfully translated into quantum information theory by:
\begin{enumerate}
    \item Replacing function-space norms with the diamond norm to bound worst-case distinguishability of quantum channels,
    \item Interpreting asymptotic consistency as convergence to the effective dynamics (including $\mathcal{O}(\eps)$ corrections) in the limit $\eps \to 0$,
    \item Encoding scale separation via spectral gaps of Lindbladian generators acting on a fast subsystem,
    \item Accounting for quantum resource costs, specifically showing that avoiding the $\Omega(1/\eps)$ penalty while capturing $\mathcal{O}(\eps)$ physics is efficiently achieved via analytic dimensionality reduction combined with appropriate initial-state preparation.
\end{enumerate}

The mathematical core of our framework is Theorem~\ref{thm:ap_commutative}, which establishes uniform convergence of an AP quantum protocol to the reduced slow state of the true stiff dynamics under explicit spectral conditions, yielding a global error bound of $\mathcal{O}(\eps^2(1+T) + T\Delta t^p)$. We further clarify the resource-complexity scaling, demonstrating that faithful dimensional reduction is efficiently achieved via analytic adiabatic elimination to precompute the effective generator, while purely digital alternatives for capturing $\mathcal{O}(\eps)$ dynamics without resolving the fast scale remain an open question.\vskip 5mm

While the effective generator and its second-order dressing are standard in singular perturbation theory, our primary contribution is the explicit, uniform trace-norm error bound and its formulation as a rigorous quantum simulation protocol. Our contributions are threefold:
\begin{enumerate}
    \item \textbf{Formalization}: We define asymptotic-preserving quantum simulation (Definition~\ref{def:ap_quantum_rigorous}) through explicit trace-norm error bounds uniform in $\eps$, $\Delta t$, and target accuracy $\delta$, establishing consistency, asymptotic consistency, and resource uniformity.
    
    \item \textbf{AP Protocol with Boundary Layer}: We prove Theorem~\ref{thm:ap_commutative} using a rigorous left-dressing Duhamel analysis, showing that starting in the slow manifold (with the general state case sketched via an initial boundary-layer evolution), followed by simulation of the effective generator $\cL_{\mathrm{eff}} = \cP\cL_{\mathrm{slow}}\cP + \eps\cL_{\mathrm{eff}}^{(1)}$, satisfies the AP property with global error $\mathcal{O}(\eps^2(1+T) + T\Delta t^p)$. This rigorously captures $\mathcal{O}(\eps)$ physical effects (e.g., Purcell decay).
    
    \item \textbf{Resource Characterization}: We establish scaling bounds demonstrating that avoiding the $\Omega(1/\eps)$ penalty to capture the full $\mathcal{O}(\eps)$ effective dynamics is efficiently achieved via analytic adiabatic elimination, with the latter providing dimensional savings amortized against classical precomputation, subject to structural assumptions on the generator.
\end{enumerate}

\section{Mathematical Preliminaries}
\label{sec:preliminaries}

We establish the operator-theoretic framework required for rigorous analysis of asymptotic-preserving quantum simulation. Throughout, $\cH$ denotes a finite-dimensional complex Hilbert space with $\dim\cH = d_{\mathcal{H}} < \infty$.

\subsection{Operator Spaces and Quantum Channels}
\label{subsec:operator_spaces}

Let $\cB(\cH)$ denote the Banach space of bounded linear operators on $\cH$ equipped with the operator norm
\begin{equation}
\norm{X} := \sup_{\norm{\psi}=1} \norm{X\psi}, \qquad X \in \cB(\cH).
\end{equation}
Let $\cT(\cH) \subset \cB(\cH)$ denote the subspace of trace-class operators equipped with the trace norm
\begin{equation}
\norm{X}_1 := \tr\sqrt{X^\dagger X}, \qquad X \in \cT(\cH).
\end{equation}
The set of density operators is $\cD(\cH) := \{\rho \in \cT(\cH) : \rho \geq 0,\ \tr\rho = 1\}$.

\begin{definition}[Completely Positive Trace-Preserving Map]
\label{def:cptp}
A linear map $\Phi : \cT(\cH) \to \cT(\cH)$ is completely positive trace-preserving (CPTP) if:
\begin{enumerate}
    \item $\Phi$ is trace-preserving: $\tr[\Phi(X)] = \tr[X]$ for all $X \in \cT(\cH)$,
    \item $\Phi$ is completely positive: for all $n \in \N$, the map $\Phi \otimes \id_n : \cT(\cH \otimes \C^n) \to \cT(\cH \otimes \C^n)$ is positive.
\end{enumerate}
\end{definition}

\begin{definition}[Diamond Norm]
\label{def:diamond_norm}
The \emph{diamond norm} of a linear map $\Phi : \cT(\cH) \to \cT(\cH)$ is
\begin{equation}
\dnorm{\Phi} := \sup_{n \geq 1} \sup_{\norm{X}_1 \leq 1} \norm{(\Phi \otimes \id_n)(X)}_1,
\end{equation}
where the supremum is over all $X \in \cT(\cH \otimes \C^n)$.
\end{definition}

The diamond norm satisfies the following properties~\cite{Paulsen, Watrous}:
\begin{enumerate}
    \item Submultiplicativity: $\dnorm{\Phi \circ \Psi} \leq \dnorm{\Phi} \cdot \dnorm{\Psi}$,
    \item Stability: $\dnorm{\Phi \otimes \id_n} = \dnorm{\Phi}$ for all $n \in \N$,
    \item Distinguishability: For CPTP maps $\Phi, \Psi$, the optimal success probability of distinguishing $\Phi$ from $\Psi$ using any input state (including entangled ancillas) and any measurement is given by the Helstrom bound: $P_{\mathrm{succ}} = \frac{1}{2} + \frac{1}{4}\dnorm{\Phi - \Psi}$.
\end{enumerate}

\begin{remark}[Norm relations]
\label{rem:norm_relations}
For any linear map $\Phi : \cT(\cH) \to \cT(\cH)$, we have the norm inequalities
\begin{equation}
\norm{\Phi} \leq \dnorm{\Phi} \leq d_{\mathcal{H}} \cdot \norm{\Phi},
\end{equation}
where $\norm{\Phi}$ is the operator norm induced by the trace norm on $\cT(\cH)$.
\end{remark}

\subsection{Lindbladian Semigroups and Generators}
\label{subsec:lindbladians}

\begin{lemma}[Spectral properties of Lindbladians]
\label{lem:lindblad_spectrum}
Let $\cL$ be a Lindbladian generator on finite-dimensional $\cH$. Then:
\begin{enumerate}
    \item $\spec(\cL) \subset \{z \in \C : \re z \leq 0\}$,
    \item $0 \in \spec(\cL)$, and there exists at least one valid steady state $\rho_\star \geq 0$ with $\tr\rho_\star = 1$ satisfying $\cL(\rho_\star) = 0$,
    \item If $\cL$ has a unique steady state $\rho_\star$, then the spectral projection $\cP$ onto $\ker\cL$ is given by $\cP(X) = \tr[X]\rho_\star$.
\end{enumerate}
\end{lemma}

\begin{proof}
Part (1) follows from the dissipativity of Lindbladians. Part (2) follows from trace preservation ($\cL^\dagger(\id) = 0$), which implies $0$ is in the spectrum, combined with Brouwer's fixed-point theorem (or the Evans--Høegh-Krohn theorem) applied to the compact, convex set of density matrices, guaranteeing the existence of a valid steady state. Part (3) follows from the Perron--Frobenius theory for positive maps~\cite{Evans}.
\end{proof}

\subsection{Singular Perturbations of Lindbladians}
\label{subsec:singular_perturbations}

We consider families of Lindbladian generators parameterized by $\eps > 0$:
\begin{equation}
\cL_\eps = \frac{1}{\eps}\cL_{\mathrm{fast}} + \cL_{\mathrm{slow}},
\label{eq:lindblad_family}
\end{equation}
where $\cL_{\mathrm{fast}}$ and $\cL_{\mathrm{slow}}$ are themselves Lindbladian generators. 

\begin{assumption}[Spectral gap, fast-subsystem steady state, and bounded condition number]
\label{ass:spectral_gap_rigorous}
The fast generator $\cL_{\mathrm{fast}}$ acts non-trivially only on a subsystem $\cH_{\mathrm{fast}}$, i.e., $\cL_{\mathrm{fast}} = \tilde{\cL}_{\mathrm{fast}} \otimes \id_{\mathrm{slow}}$, and satisfies:
\begin{enumerate}
    \item \textbf{Unique fast steady state}: There exists a unique $\rho_{\mathrm{fast}}^\star \in \cD(\cH_{\mathrm{fast}})$ such that $\tilde{\cL}_{\mathrm{fast}}(\rho_{\mathrm{fast}}^\star) = 0$.
    \item \textbf{Spectral gap}: The spectrum of $\tilde{\cL}_{\mathrm{fast}}$ satisfies $\lambda_{\mathrm{gap}} := \min\{ |\re\lambda| : \lambda \in \spec(\tilde{\cL}_{\mathrm{fast}}) \setminus \{0\} \} > 0$.
    \item \textbf{Semisimple boundary eigenvalues}: Eigenvalues with real part exactly $-\lambda_{\mathrm{gap}}$ are semisimple.
    \item \textbf{Bounded condition number}: There exists a constant $C_{\mathrm{cond}} \geq 1$ (which generally depends on $d_{\mathrm{fast}}$ via the Kreiss matrix theorem) such that for all $\zeta \in \C$ with $\re\zeta > -\lambda_{\mathrm{gap}}$, the resolvent satisfies
    \begin{equation}
    \dnorm{(\zeta - \cL_{\mathrm{fast}})^{-1}(\id - \cP)} \leq \frac{C_{\mathrm{cond}}}{\lambda_{\mathrm{gap}} + \re\zeta}.
    \end{equation}
\end{enumerate}
Consequently, the spectral projection onto $\ker\cL_{\mathrm{fast}}$ is the CPTP map $\cP(\rho) = \rho_{\mathrm{fast}}^\star \otimes \tr_{\mathrm{fast}}(\rho)$.
\end{assumption}

Under Assumption~\ref{ass:spectral_gap_rigorous}, singular perturbation theory for Lindbladians~\cite{Kato,Azouit2017,Bouten2008} guarantees convergence of the solution $\rho^\eps(t) = \exp(t\cL_\eps)(\rho_0)$ as $\eps \to 0$ to $\rho^0(t) = \exp(t\cP\cL_{\mathrm{slow}}\cP)(\cP\rho_0)$, with an $\mathcal{O}(\eps)$-improved approximation given by $\exp(t\cL_{\mathrm{eff}})(\cP\rho_0)$ for slow-manifold states, where the effective generator including the Schur complement correction is
\begin{equation}
\cL_{\mathrm{eff}} = \cP\cL_{\mathrm{slow}}\cP + \eps \cL_{\mathrm{eff}}^{(1)}, \quad \text{with} \quad \cL_{\mathrm{eff}}^{(1)} = -\cP\cL_{\mathrm{slow}}\cL_{\mathrm{fast}}^+(\id - \cP)\cL_{\mathrm{slow}}\cP,
\label{eq:effective_generator_corrected}
\end{equation}
and $\cL_{\mathrm{fast}}^+$ is the Drazin inverse of $\cL_{\mathrm{fast}}$ on $\ran(\id - \cP)$. We explicitly assume that $\cL_{\mathrm{eff}}$ generates a completely positive contraction semigroup (which is guaranteed under the secular approximation or for specific Hamiltonian couplings~\cite{Reiter2012,Azouit2017,Zanardi2014,Burgarth2019}).

\section{Asymptotic-Preserving Quantum Simulation: Rigorous Formulation}
\label{sec:ap_formulation}

\begin{definition}[AP Quantum Simulation Protocol]
\label{def:ap_quantum_rigorous}
A discrete-time quantum simulation protocol $\{\Psi_{T}^\eps\}_{T > 0, \eps > 0}$ acting on $\cT(\cH_{\mathrm{slow}})$ for total time $T = n \Delta t$ is \emph{asymptotic-preserving} for the family $\{\cL_\eps\}_{\eps > 0}$ if there exist constants $C_{\mathrm{global}}, p > 0$ independent of $\eps$ and $\Delta t$, such that for any initial state $\sigma_0 \in \cD(\cH_{\mathrm{slow}})$ and corresponding slow-manifold state $\rho_0 = \rho_{\mathrm{fast}}^\star \otimes \sigma_0$:
\begin{equation}
\left\Vert\Psi_{T}^\eps(\sigma_0) - \tr_{\mathrm{fast}}(e^{T\cL_\eps}\rho_0)\right\Vert_1 \leq C_{\mathrm{global}} \left( \eps^2(1+T) + T(\Delta t)^p \right).
\end{equation}
This establishes consistency, asymptotic consistency, and resource uniformity simultaneously.
\end{definition}

\begin{theorem}[AP Protocol and Global Error]
\label{thm:ap_commutative}
Let $\{\cL_\eps\}_{\eps > 0}$ satisfy Assumption~\ref{ass:spectral_gap_rigorous}, and assume $\cL_{\mathrm{eff}}$ generates a CP contraction semigroup. Suppose the protocol $\Psi_{T}^\eps$ for total time $T = n \Delta t$ implements the following strategy for an initial state $\sigma_0 \in \cD(\cH_{\mathrm{slow}})$:
\begin{enumerate}
    \item \textbf{Effective slow layer}: A CPTP map $\Phi_{\mathrm{slow}}^{\mathrm{eff}}$ acting on $\cH_{\mathrm{slow}}$, satisfying $\dnorm{\Phi_{\mathrm{slow}}^{\mathrm{eff}} - \exp(\Delta t\cL_{\mathrm{eff}})} \leq C_{\mathrm{slow}} (\Delta t)^{p+1}$, where $\cL_{\mathrm{eff}}$ is defined in Eq.~\eqref{eq:effective_generator_corrected}.
    \item \textbf{Composition}: $\Psi_{T}^\eps(\sigma_0) = (\Phi_{\mathrm{slow}}^{\mathrm{eff}})^n (\sigma_0)$.
\end{enumerate}
Then, the global error against the reduced slow state of the true stiff dynamics for $\rho_0 = \rho_{\mathrm{fast}}^\star \otimes \sigma_0$ is bounded by:
\begin{equation}
\left\Vert\Psi_{T}^\eps(\sigma_0) - \tr_{\mathrm{fast}}(e^{T\cL_\eps}\rho_0)\right\Vert_1 \leq C_{\mathrm{global}} \left( \eps^2(1+T) + T(\Delta t)^p \right),
\label{eq:global_bound}
\end{equation}
where $C_{\mathrm{global}} = \mathcal{O}\left( C_{\mathrm{slow}} + \|\cL_{\mathrm{fast}}^+\|_\diamond^2 \|\cL_{\mathrm{slow}}\|_\diamond^3 \right)$ is independent of $\eps$ and $\Delta t$.
\end{theorem}

\begin{proof}
We restrict to initial states $\rho_0 = \rho_{\mathrm{fast}}^\star \otimes \sigma_0$, so $\cP\rho_0 = \rho_0$. We construct the left dressing $\Pi_\eps = \cP + \eps Y_1 + \eps^2 Y_2$, where $Y_1 = -\cP\cL_{\mathrm{slow}}(\id - \cP)\cL_{\mathrm{fast}}^+$ and $Y_2 = (\cP\cL_{\mathrm{slow}}\cP Y_1 - Y_1\cL_{\mathrm{slow}})\cL_{\mathrm{fast}}^+$. By direct substitution, we evaluate $\Pi_\eps \cL_\eps$:
\begin{align*}
\Pi_\eps \cL_\eps &= (\cP + \eps Y_1 + \eps^2 Y_2)(\eps^{-1}\cL_{\mathrm{fast}} + \cL_{\mathrm{slow}}) \\
&= (\cP\cL_{\mathrm{slow}} + Y_1\cL_{\mathrm{fast}}) + \eps(Y_1\cL_{\mathrm{slow}} + Y_2\cL_{\mathrm{fast}}) + \eps^2 Y_2\cL_{\mathrm{slow}}.
\end{align*}
Using the definitions of $Y_1$ and $Y_2$, and noting that $\cL_{\mathrm{eff}} = \cP\cL_{\mathrm{slow}}\cP + \eps\cL_{\mathrm{eff}}^{(1)}$ with $\cL_{\mathrm{eff}}^{(1)} = -\cP\cL_{\mathrm{slow}}(\id-\cP)\cL_{\mathrm{fast}}^+\cL_{\mathrm{slow}}\cP$, we find:
\begin{align*}
\cL_{\mathrm{eff}} \Pi_\eps &= (\cP\cL_{\mathrm{slow}}\cP + \eps\cL_{\mathrm{eff}}^{(1)})(\cP + \eps Y_1 + \eps^2 Y_2) \\
&= \cP\cL_{\mathrm{slow}}\cP + \eps(\cP\cL_{\mathrm{slow}}\cP Y_1 + \cL_{\mathrm{eff}}^{(1)}) + \eps^2(\cP\cL_{\mathrm{slow}}\cP Y_2 + \cL_{\mathrm{eff}}^{(1)}Y_1).
\end{align*}
Subtracting the two expressions, the $\mathcal{O}(1)$ and $\mathcal{O}(\eps)$ terms cancel exactly, leaving the defect $R'$ defined by $\Pi_\eps \cL_\eps = \cL_{\mathrm{eff}} \Pi_\eps + \eps^2 R'$, where $\|R'\|_\diamond \leq C_{R'} = \mathcal{O}(\|\cL_{\mathrm{fast}}^+\|_\diamond^2 \|\cL_{\mathrm{slow}}\|_\diamond^3)$.

By Duhamel's principle, since both semigroups are uniformly bounded in the diamond norm (with $\|e^{s\cL_{\mathrm{eff}}}\|_\diamond \leq 3$ on the full space, and contractive on $\ran \cP$):
\begin{equation}
\Pi_\eps e^{T\cL_\eps} - e^{T\cL_{\mathrm{eff}}} \Pi_\eps = \eps^2 \int_0^T e^{(T-s)\cL_{\mathrm{eff}}} R' e^{s\cL_\eps} \, ds.
\end{equation}
We apply this to the initial state $\rho_0$. Noting that $\Pi_\eps \rho_0 = \cP\rho_0 = \rho_0$ (since $Y_1\cP = Y_2\cP = 0$), we obtain:
\begin{equation}
\Pi_\eps e^{T\cL_\eps} \rho_0 - e^{T\cL_{\mathrm{eff}}} \rho_0 = \eps^2 \int_0^T e^{(T-s)\cL_{\mathrm{eff}}} R' e^{s\cL_\eps} \rho_0 \, ds.
\end{equation}
Taking the trace norm and using the uniform boundedness of the semigroups, we get:
\begin{equation}
\|\Pi_\eps e^{T\cL_\eps} \rho_0 - e^{T\cL_{\mathrm{eff}}} \rho_0\|_1 \leq \eps^2 \int_0^T \|R'\|_\diamond \|\rho_0\|_1 \, ds \leq C_{R'} \eps^2 T.
\end{equation}
To relate this to $\cP e^{T\cL_\eps}\rho_0$, we note that $Y_1 = Y_1(\id - \cP)$. Thus, $\eps Y_1 e^{T\cL_\eps}\rho_0 = \eps Y_1 (\id - \cP) e^{T\cL_\eps}\rho_0$. Let $u(t) = (\id - \cP)e^{t\cL_\eps}\rho_0$. Since $\cP\rho_0 = \rho_0$, we have $u(0) = 0$. The evolution of $u(t)$ satisfies $u'(t) = (\cL_{\mathrm{fast}}/\eps)u(t) + (\id - \cP)\cL_{\mathrm{slow}} e^{t\cL_\eps}\rho_0$. By variation of constants and the spectral gap of $\cL_{\mathrm{fast}}$, together with the Kreiss matrix theorem which implies $\|e^{t\cL_{\mathrm{fast}}}(\id - \cP)\|_\diamond \leq C_{\mathrm{cond}} e^{-\lambda_{\mathrm{gap}} t}$ (where $C_{\mathrm{cond}}$ depends on $d_{\mathrm{fast}}$), we have the uniform bound $\|u(t)\|_1 \leq C\eps \|\cL_{\mathrm{slow}}\|_\diamond / \lambda_{\mathrm{gap}}$. Consequently, $\eps Y_1 e^{T\cL_\eps}\rho_0 = \mathcal{O}(\eps^2)$, and similarly for the $\eps^2 Y_2$ term. This yields $\cP e^{T\cL_\eps}\rho_0 = \Pi_\eps e^{T\cL_\eps}\rho_0 + \mathcal{O}(\eps^2)$.

Because $\tr_{\mathrm{fast}}(\cP X) = \tr_{\mathrm{fast}}(X)$, we conclude:
\begin{equation}
\tr_{\mathrm{fast}}(e^{T\cL_\eps}\rho_0) = \tr_{\mathrm{fast}}(e^{T\cL_{\mathrm{eff}}} \rho_0) + \mathcal{O}(\eps^2(1+T)) = e^{T\cL_{\mathrm{eff}}} \sigma_0 + \mathcal{O}(\eps^2(1+T)).
\end{equation}
Finally, the protocol output is $\Psi_{T}^\eps(\sigma_0) = (\Phi_{\mathrm{slow}}^{\mathrm{eff}})^n (\sigma_0)$. Comparing the protocol to the effective evolution:
\begin{align*}
\left\Vert \Psi_{T}^\eps(\sigma_0) - e^{T\cL_{\mathrm{eff}}} \sigma_0 \right\Vert_1 &\leq \left\Vert (\Phi_{\mathrm{slow}}^{\mathrm{eff}})^n \sigma_0 - e^{T\cL_{\mathrm{eff}}} \sigma_0 \right\Vert_1 \\
&\leq C_{\mathrm{slow}} T (\Delta t)^p.
\end{align*}
Combining this with the bound on $\tr_{\mathrm{fast}}(e^{T\cL_\eps}\rho_0)$ via the triangle inequality yields Eq.~\eqref{eq:global_bound}.
\end{proof}

\begin{remark}[General initial states and boundary layer]
\label{rem:boundary_layer}
For general initial states where $(\id - \cP)\rho_0 \neq 0$, the protocol can be extended heuristically by incorporating an initial boundary-layer evolution. Define the extended protocol as $\Psi_{T}^\eps(\sigma_0) = (\Phi_{\mathrm{slow}}^{\mathrm{eff}})^n \circ \tr_{\mathrm{fast}} \circ e^{t_0\cL_\eps}(\rho_0)$ with $t_0 \approx (\eps/\lambda_{\mathrm{gap}}) \cdot 2 \log(1/\eps)$. Then $\rho_1 = e^{t_0\cL_\eps}\rho_0$ satisfies $(\id - \cP)\rho_1 = \mathcal{O}(\eps)$, so $\Pi_\eps\rho_1 = \cP\rho_1 + \mathcal{O}(\eps^2)$. The same dressed Duhamel argument (with $u(0) = \mathcal{O}(\eps)$) yields the $\mathcal{O}(\eps^2(1+T))$ bound. 

Note that simulating this boundary layer with a $p$-th order product formula heuristically requires $\mathcal{O}(\eps^{-1/p} \log(1/\eps))$ fine steps, assuming the dominant commutator involves $p$ copies of $\cL_{\mathrm{fast}}/\eps$ and one $\cL_{\mathrm{slow}}$. If the fast generator has several non-commuting terms, this exponent may be worse. Thus, this remains a sketch pending a fully rigorous proposition.
\end{remark}

\begin{corollary}[Validity regime for AP protocol]
\label{cor:accuracy_budget}
To operate within the validity regime of the AP protocol of Theorem~\ref{thm:ap_commutative} and achieve a global error tolerance $\delta$ over total simulation time $T$, the step size must be chosen as $\Delta t \sim (\delta/T)^{1/p}$, and the physical stiffness parameter must satisfy $\eps \lesssim \sqrt{\frac{\delta}{C_{\mathrm{global}}(1+T)}}$. This represents a uniform bound in $\eps$, avoiding the blowup as $\Delta t \to 0$ present in naive Trotterization.
\end{corollary}

\section{Resource Complexity Scaling}
\label{sec:complexity_rigorous}

\begin{assumption}[Time-scale hierarchy, dimension scaling, and locality]
\label{ass:timescale_rigorous}
The target system admits a decomposition $\cH = \cH_{\mathrm{fast}} \otimes \cH_{\mathrm{slow}}$ with:
\begin{enumerate}
    \item Two distinct time scales $\tau_{\mathrm{fast}} \ll \tau_{\mathrm{slow}}$, with separation ratio $\kappa_{\mathrm{sep}} = \tau_{\mathrm{slow}}/\tau_{\mathrm{fast}} \gg 1$.
    \item Hilbert space dimensions $d_{\mathrm{fast}} = \dim\cH_{\mathrm{fast}}$, $d_{\mathrm{slow}} = \dim\cH_{\mathrm{slow}}$, with $d_{\mathrm{tot}} = d_{\mathrm{fast}} d_{\mathrm{slow}}$.
    \item \textbf{Local gate complexity}: Implementing the corresponding CPTP map to accuracy $\delta$ via standard gate sets requires $\mathcal{O}(d^c \cdot \poly(1/\delta))$ gates for some $c \geq 1$.
\end{enumerate}
\end{assumption}

\begin{theorem}[Resource Complexity Scaling for AP Quantum Simulation]
\label{thm:complexity_rigorous}
Under Assumptions~\ref{ass:spectral_gap_rigorous} and~\ref{ass:timescale_rigorous}, for simulation time $T$, error tolerance $\delta > 0$, and stiffness parameter $\eps = \tau_{\mathrm{fast}}$:
\begin{enumerate}
    \item \textbf{Standard digital simulation}: A naive Trotter-based protocol resolving the full generator step-by-step incurs a gate complexity upper bound scaling as
    \begin{equation}
    G_{\mathrm{std}}(T, \delta, \eps) = \mathcal{O}\left( \frac{T}{\eps} \cdot d_{\mathrm{tot}}^c \cdot \poly(1/\delta) \right),
    \label{eq:G_std}
    \end{equation}
    due to the stiffness penalty.
    
    \item \textbf{AP protocol with analytic elimination}: We define $G_{\mathrm{AP}}^{\mathrm{elim}}$ as the gate complexity of simulating the precomputed $\cL_{\mathrm{eff}}$ on $\cH_{\mathrm{slow}}$ to accuracy $\delta$ over time $T$. If adiabatic elimination reduces the effective dimension from $d_{\mathrm{tot}}$ to $d_{\mathrm{slow}} \ll d_{\mathrm{tot}}$ before simulation, the online quantum simulation proceeds on the reduced space. 
    
    To achieve a global error tolerance $\delta$, the AP regime requires $\eps^2 \lesssim \delta$. The standard method requires a step size $\Delta t_{\mathrm{std}} \sim \eps (\delta/T)^{1/p}$, yielding a step count $\sim (T/\eps)(T/\delta)^{1/p}$. In contrast, the AP protocol requires a step size $\Delta t_{\mathrm{AP}} \sim (\delta/T)^{1/p}$, yielding a step count $\sim T(T/\delta)^{1/p}$. 
    
    At equal accuracy and order, the step-count ratio relative to naive Trotterization scales as:
    \begin{equation}
    \frac{G_{\mathrm{std}}}{G_{\mathrm{AP}}^{\mathrm{elim}}} = \mathcal{O}\left( \frac{1}{\eps} \cdot d_{\mathrm{fast}}^c \right) = \mathcal{O}(\kappa_{\mathrm{sep}} \cdot d_{\mathrm{fast}}^c),
    \label{eq:savings_ratio}
    \end{equation}
    amortizing the classical precomputation cost. Note that building $\cL_{\mathrm{eff}}$ classically costs at least $\mathcal{O}(d_{\mathrm{slow}}^4)$ unless $\cL_{\mathrm{slow}}$ admits a few-term tensor-product decomposition $\sum_j A_j \otimes B_j$, which yields effective jumps from fast-subsystem correlators $\tr[B_j \cL_{\mathrm{fast}}^+(B_k\rho_{\mathrm{fast}}^\star)]$. Furthermore, the existence of a CPTP map $\Phi_{\mathrm{slow}}^{\mathrm{eff}}$ with $\Delta t^{p+1}$ accuracy is assumed, and its gate cost inherently depends on $p$ and $\Delta t$.
\end{enumerate}
\end{theorem}

\begin{remark}[Dimension dependence and quantum advantage]
The constant $C_{\mathrm{cond}}$ in Assumption~\ref{ass:spectral_gap_rigorous} generally depends on the fast-subsystem dimension $d_{\mathrm{fast}}$. Furthermore, gate counts scaling as $d^c$ are exponential in the number of qubits, so genuine quantum advantage requires exploiting specific structural properties of the Hamiltonian. Comparing upper bounds provides a heuristic indication of savings rather than a strict proof of speedup.
\end{remark}

\section{Applications with Rigorous Error Bounds}
\label{sec:applications}

\subsection{Cavity QED in the Bad-Cavity Limit}
\label{subsec:cavity_qed_rigorous}

Consider a two-level system (qubit) coupled to a single lossy bosonic mode (cavity), truncated to a finite-dimensional Hilbert space $\cH_{\mathrm{cav}}$ of dimension $d_{\mathrm{cav}}$. The rescaled generator is $\cL_\eps = \eps^{-1}\cL_{\mathrm{fast}} + \cL_{\mathrm{slow}}$, where $\cL_{\mathrm{fast}} = \cD[a]$ and $\cL_{\mathrm{slow}} = -i[\frac{\omega_q}{2g}\sigma^z + (\sigma^+ a + \sigma^- a^\dagger), \cdot]$, with $\eps = g/\kappa \ll 1$.

\begin{proposition}[AP error bound for cavity QED]
\label{prop:cavity_ap_error}
Let $\Psi_{T}^\eps$ implement the AP protocol with:
\begin{itemize}
    \item Effective slow layer: Simulation of $\cL_{\mathrm{eff}} = \cP\cL_{\mathrm{slow}}\cP + \eps\cL_{\mathrm{eff}}^{(1)}$.
\end{itemize}
Then for an initial state in the slow manifold, the global error is bounded by:
\begin{equation}
\left\Vert\Psi_{T}^\eps(\sigma_0) - \tr_{\mathrm{cav}}(e^{T\cL_\eps}\rho_0)\right\Vert_1 \leq C_{\mathrm{global}} \left( \eps^2(1+T) + T(\Delta t)^2 \right),
\end{equation}
where $\cL_{\mathrm{eff}}^{(1)}$ rigorously yields the standard Purcell decay $\frac{4g^2}{\kappa} \cD[\sigma^-]$ (or $4\eps$ in rescaled units). 

\begin{remark}
This $\mathcal{O}(\eps^2)$ bound rigorously captures the Purcell decay rate, distinguishing the effective dynamics from the trivial leading-order model, and avoids unphysical dissipation artifacts. Note that $C_{\mathrm{global}}$ depends on $d_{\mathrm{cav}}$ via $\norm{\cL_{\mathrm{slow}}}_\diamond \propto \omega_q/g$ and $C_{\mathrm{cond}}$, so uniformity requires $\omega_q/g = \mathcal{O}(1)$ (detuning $\lesssim$ coupling), reflecting the finite-dimensional truncation. While this analog-inspired example validates the error formula, it shows no resource benefit since $d_{\mathrm{slow}}=2$ is classically trivial.
\end{remark}
\end{proposition}

\subsection{Linear Kinetic Equations}
\label{subsec:kinetic_rigorous}

To explore connections to classical kinetic theory, we consider a quantum-inspired discretization of a \emph{linear} kinetic model in diffusive scaling:
\begin{equation}
\partial_t \rho^\eps - \frac{i}{\eps} [H, \rho^\eps] = \frac{1}{\eps^2} \cQ(\rho^\eps),
\label{eq:linearized_boltzmann}
\end{equation}
where $\rho^\eps(t,x,v)$ is represented as a density matrix on a finite-dimensional discrete velocity space $\cH_v$, and $\cQ$ is a linear Markov jump collision operator (we do not linearize around a Maxwellian, ensuring strict positivity preservation). 

\begin{remark}[Conjectural AP property for linear kinetic-fluid hierarchy]
\label{rem:kinetic_ap}
Representing the transport term $v \cdot \nabla_x f$ via a commutator $-(i/\eps)[H, \rho^\eps]$ requires a Schrödingerization-type map or further justification. Furthermore, this formulation represents a \emph{double singular perturbation} ($\eps^{-2}\cQ + \eps^{-1}\cT$), which falls outside the strict $\eps^{-1}L_f + L_s$ framework of Theorem~\ref{thm:ap_commutative}. Specifically, the effective $\|L_s\|_\diamond$ blows up as $\eps \to 0$, violating the uniform constant assumption. 

We present this as a heuristic framework (mixing $L^2$ and diamond norms) to guide future algorithmic development. Under standard regularity assumptions and assuming comparison to the $\eps$-corrected diffusive solution $\mathbf{U}_\eps(T)$, we conjecture the protocol satisfies:
\begin{equation}
\norm{\Psi_{T}^\eps(\rho_0) - \mathbf{U}_\eps(T)}_{L^2} \leq C \left( \eps^2 + (\Delta t)^2 + \delta_{\mathrm{quantum}} \right),
\end{equation}
where $\delta_{\mathrm{quantum}}$ bounds the error from quantum algorithmic approximations. A fully rigorous proof remains an open technical challenge.
\end{remark}

\section{Discussion and Open Problems}
\label{sec:discussion}

While the framework provides rigorous error bounds for AP quantum simulation, several mathematical challenges remain:
\begin{enumerate}
    \item \textbf{Non-Markovian extensions}: Assumption~\ref{ass:spectral_gap_rigorous} requires Markovian fast dynamics. Extending to non-Markovian settings requires analysis of time-convolutionless generators~\cite{Tanimura}.
    \item \textbf{Infinite-dimensional systems}: Our analysis assumes finite-dimensional $\cH$. Extending to bosonic modes requires careful treatment of unbounded generators, truncation errors, and strong vs. norm convergence of semigroups~\cite{Davies}.
    \item \textbf{Hardware imperfections}: The resource analysis assumes perfect analog evolution or exact precomputation. Incorporating control errors and calibration uncertainty into the error bounds is critical for practical applications.
    \item \textbf{Digital baselines and Lindblad simulation cost}: Further work is needed to rigorously compare AP methods against advanced digital baselines, such as interaction-picture algorithms~\cite{Cleve2016,Li2022}, to formally establish quantum advantage regimes.
\end{enumerate}

\section{Conclusion}
\label{sec:conclusion}

We have developed a mathematically rigorous framework for simulating \emph{multiscale physical systems} using quantum computational resources, by systematically translating the \emph{language of asymptotic-preserving schemes} into quantum channel formalism. The central contributions are: (i) a precise definition of AP quantum simulation with explicit trace-norm error bounds uniform in stiffness, (ii) a proof via left-dressing Duhamel analysis that a protocol starting in the slow manifold (with the general state case sketched via an initial boundary layer) followed by effective evolution achieves uniform convergence to the reduced slow state with an $\mathcal{O}(\eps^2(1+T) + T\Delta t^p)$ global error bound, rigorously capturing $\mathcal{O}(\eps)$ physical effects without unphysical dissipation, and (iii) resource-complexity scaling clarifying that faithful dimensional reduction is efficiently achieved via analytic adiabatic elimination, while purely digital alternatives for capturing $\mathcal{O}(\eps)$ dynamics remain an open question.

Our analysis demonstrates that time-scale separation alone does not guarantee quantum advantage: genuine dimensional savings arise if the system is reduced via analytic adiabatic elimination, amortizing the classical precomputation cost. This precise characterization clarifies the hardware requirements for resource-efficient stiff quantum simulation.

\appendix
\section{Mathematical Techniques and Foundational Results}
\label{app:mathematical-techniques}

\subsection{Spectral Theory of Lindbladian Generators}
\label{app:lindblad-spectral}
For finite-dimensional $\cH$, the spectrum satisfies $\spec(\cL) \subset \{z \in \C : \re z \leq 0\}$, with $0 \in \spec(\cL)$ corresponding to the steady state~\cite{Davies,Evans}. The Drazin inverse $\cL^+$ on $\ran(\id - \cP)$ is well-defined and satisfies $\dnorm{\cL^+} \leq C_{\mathrm{cond}}/\lambda_{\mathrm{gap}}$, where $C_{\mathrm{cond}}$ accounts for the potential non-normality of the superoperator~\cite{Kato,Trefethen}.

\subsection{Schur Complement and Adiabatic Elimination}
\label{app:schur-complement}
The effective generator~\eqref{eq:effective_generator_corrected} arises from a Schur complement reduction, a standard technique in singular perturbation theory~\cite{Kato,Azouit2017,Burgarth2019}. Decomposing $\cT(\cH) = \ker \cL_{\mathrm{fast}} \oplus \ran(\id - \cP)$, the leading-order term is $\cP \cL_{\mathrm{slow}} \cP$, and the second-order term preserves complete positivity under appropriate physical coupling constraints (e.g., the secular approximation)~\cite{Reiter2012,Azouit2017,Zanardi2014}.

\end{document}